\documentclass[conference]{IEEEtran}

\usepackage{graphicx}
\usepackage{amsmath,epsfig,amssymb,mathrsfs}
\usepackage{algorithmic}
\usepackage{cite}
\usepackage{amsthm}
\usepackage{algorithm}
\usepackage{arydshln}
\usepackage{amsopn}
\usepackage{color}
\usepackage{float}
\floatstyle{plaintop}
\restylefloat{table}

\usepackage{etoolbox}
\makeatletter
\patchcmd{\@makecaption}
  {\scshape}
  {}
  {}
  {}
\makeatother

\newtheorem{theorem}{Theorem}[section]
\newtheorem{lemma}[theorem]{Lemma}

\newtheorem{proposition}[theorem]{Proposition}
\newtheorem{corollary}[theorem]{Corollary}

\usepackage{amsopn}

\usepackage{amssymb}

\newcommand{\E}{\mathbb{E}}

\newcommand\independent{\protect\mathpalette{\protect\independenT}{\perp}}
\def\independenT#1#2{\mathrel{\rlap{$#1#2$}\mkern2mu{#1#2}}}

\IEEEoverridecommandlockouts

\begin{document}
\title{An Information-Theoretic Measure of Dependency Among Variables in Large Datasets}
\author{\IEEEauthorblockN{Ali Mousavi, Richard G. Baraniuk}
\IEEEauthorblockA{Department of Electrical and Computer Engineering \\
Rice University\\
Houston, TX 77005\\
}
\thanks{ This work was supported by NSF CCF-0926127, CCF-1117939; DARPA/ONR N66001-11-C-4092 and N66001-11-1-4090; ONR N00014-10-1-0989, and N00014-11-1-0714; ARO MURI W911NF-09-1-0383.

Email: \{ali.mousavi, richb\} @rice.edu }
}

\maketitle

\begin{abstract}
The maximal information coefficient (MIC),
which measures the amount of dependence between two variables,
is able to detect both linear and non-linear associations. However,
computational cost grows rapidly as a function of the dataset size.
In this paper, we develop a computationally efficient
approximation to the MIC that replaces its dynamic programming
step with a much simpler technique based on the uniform partitioning of data grid.
A variety of experiments demonstrate the quality of our approximation.
\end{abstract}

\section{Introduction}
One of the challenging issues for statisticians and computer scientists is dealing with large datasets containing hundreds of variables which some of them may have interesting but unexplored relationships with each other. This is due to
examples of massive datasets in different areas such as: social networks, astronomy, genomics, medical records, and political science. Hence, it is an interesting topic to try to come up with methods which help us to discover these relationships.

Measuring the amount of dependence among two variables has been extensively studied in the literature and several methods have been proposed for it. We review some of them in the following. In \cite{Renyi:2005}, the author has suggested seven properties to be satisfied by any measure $\phi(x,y)$ used for determining the amount of dependence between $x$ and $y$. These properties known as R\'{e}nyi's axioms are:
\begin{itemize}
 \item In defining $\phi(x,y)$ between any two random variables, neither $x$ nor $y$ should be constant with probability 1.
\item $0\leq\phi(x,y)\leq1$.
\item $\phi(x,y)=0$ if and only if $x$ and $y$ are independent from each other.
\item $\phi(x,y)=1$ if there is an arbitrary functional dependency between $x$ and $y$. In other words, if $y=f(x)$ or $x=g(y)$ where $f(.)$ and $g(.)$ are Borel-measurable functions.
\item $\phi(x,y)=\phi(y,x)$
\item if $f(.)$ and $g(.)$ are strictly monotonic functions, then $\phi(x,y)=\phi(f(x),g(y))$.
\item if $x$ and $y$ are jointly Gaussian random variables, then $\phi(x,y)=|{\rm PCC}(x,y)|$ where ${\rm PCC}$ is the Pearson correlation coefficient.
\end{itemize}

The Pearson correlation coefficient (PCC) is the most well known dependency measure. However, it is unable to detect non-linear associations. In other words, the PCC is only able to capture linear associations between two variables.

As another measure of dependency, correlation ratio of random variable $y$ (if $\sigma^2(y)$ exists and $\sigma(y)>0$) on random variable $x$, introduced in \cite{cramer:1999} and \cite{kolmogorov:1933}, is defined as
\begin{eqnarray}
 \Theta(y)=\frac{\sigma(\E(y|x))}{\sigma(y)}.
\end{eqnarray}

It is easy to show that $0\leq \Theta(y) \leq 1$ where $\Theta(y)=1$ if and only if $y=f(x)$ in which $f(x)$ is a Borel-measurable function and $\Theta(y)=0$ if $x$ and $y$ are independent. The alternative formula of the correlation ratio mentioned in \cite{Renyi:2005} is
\begin{eqnarray}
 \Theta(y) = \sup_{(g)}{|{\rm PCC}(y,g(x))|}.
\end{eqnarray}
This alternative formula leads to another measure of dependency called \textit{maximal correlation}\cite{Gabelein:1941}:
\begin{eqnarray}
  \mathbb{S}(x,y) = \sup_{f,g}{{\rm PCC}(f(x),g(y))},  
 \end{eqnarray}
where $f(.)$ and $g(.)$ are Borel-measurable functions. The author in \cite{Renyi:2005} has shown that $\mathbb{S}(x,y)=0$ if and only if $x$ and $y$ are independent. Furthermore, if there is an arbitrary functional relationship between $x$ and $y$, then $\mathbb{S}(x,y)=1$. The authors in \cite{breiman1985estimating} have introduced the \textit{alternating conditional expectation (ACE)} algorithm to find the optimal transformations.

The Spearman correlation coefficient \cite{spearman2004} is defined similar to the PCC; however, it is defined between the two ranked variables. By ranked variables we mean replacing each data point by its rank (or the average rank for equal sample points) in the ascending order. Therefore, if $\tilde{x}_i$ and $\tilde{y}_i$ denote the ranked versions of $x_i$ and $y_i$, the Spearman correlation coefficient would be
\begin{eqnarray}
 \rho=\frac{\sum_{i}{(\tilde{x}_i-\bar{x})(\tilde{y}_i-\bar{y})}}{\sqrt{\sum_{i}{(\tilde{x}_i-\bar{x})^2} \sum_{i}{(\tilde{y}_i-\bar{y})^2}}}.
\end{eqnarray}

The authors in \cite{delicado2009measuring} have expressed covariance and linear correlation in terms of principal components and generalized them for variables
distributed along a curve. They have estimated their measures using principal curves.

Mutual Information \cite{cover1991elements} is another measure that can be used for quantification of dependency between two variables since it satisfies some common properties of other dependency measures. As an example $I(x,y)=0$ if and only if $x$ and $y$ are independent.  The authors in \cite{moon1995estimation} have used kernel density estimation of probability density functions in order to estimate the mutual information between two variables. In \cite{kraskov2004estimating}, a method of mutual information estimation based on binning and estimating entropy from k-nearest neighbors is proposed.

The MIC \cite{MIC:2011} is recently proposed for quantifying dependency between two random variables. It is based on binning the dataset using dynamic programming technique to compute mutual information between different variables. It has two main properties which makes it superior in comparison with the aforementioned measures. First, it has \textit{generality} meaning that if the sample size is large enough, it is able to detect different kinds of associations rather than specific types. Second, it is an \textit{equitable} measure meaning that it gives similar scores to equally noisy associations no matter what type the association is. 

One of the problems with the MIC is the fact that its computational cost grows rapidly as a function of the dataset size. Since this computational cost may become infeasible, the authors in \cite{MIC:2011} have applied a heuristic so as not to compute the mutual information for all possible grids. This heuristic application may result in finding a local maximum.

In this paper, we develop a computationally efficient approximation to the MIC. This approximation is based on replacing the dynamic programming application used in computation of the MIC with a very efficient technique that is uniformly binning the data. We show that our proposed method is able to detect both functional and non-functional associations between different variables, similar to the MIC while more efficiently. In addition, it has a better performance in recognizing the independence between different variables.

The rest of this paper is organized as the following. In section (\ref{sec:MIC}), we review the MIC and the algorithm used to compute it from \cite{MIC:2011}. In section (\ref{sec:sub-optimal}), we introduce our new measure of dependency that is a modification to the MIC. We present simulation results in section (\ref{sec:simul}). Finally, section (\ref{sec:con}) includes the conclusion of the paper.

\section{The Maximal Information Coefficient (MIC)}\label{sec:MIC}

\subsection{MIC Definition and Properties}\label{subsec:MICDEF}
For any finite dataset $D$ which contains ordered pairs of two random variables, one can partition the first element, i.e., $x$-value of these pairs into $\ell_x$ bins and similarly partition the second element or $y$-value of these pairs into $\ell_y$ bins. As a result of this partitioning, we will have an $\ell_x$-by-$\ell_y$ grid $G$. Every cell of this grid may or may not contain some sample points from the set $D$. This grid induces a probability distribution on the cells of $G$ where the corresponding probability of each cell is equal to the portion of sample points located in that cell. That is to say
\begin{eqnarray}
p_{ij}=\frac{|D_{ij}|}{|D|},
\end{eqnarray}
where $p_{ij}$ denotes the probability corresponding to the cell located at the $i^{th}$ row and the $j^{th}$ column and $|D_{ij}|$ denotes the number of sample points falling into the $i$-th row and the $j$-th column (See Figure \ref{fig:Dij} for a graphical view of the grid G).
It is obvious that for each $(\ell_x,\ell_y)$, we will have a grid that induces a new probability distribution and hence results in a different mutual information between the two variables. 

\begin{figure}[t]
 \centering
 \includegraphics[width=60mm]{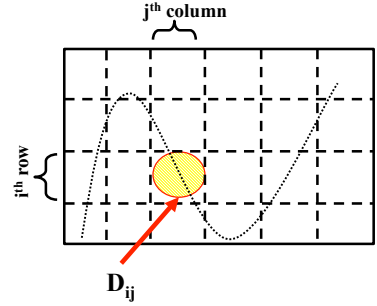}
 \caption{Partitioning of dataset $D$ into $\ell_x$ columns and $\ell_y$ rows. $D_{ij}$ denotes the set of sample points located in the $i$-th row and the $j$-th column.}
\label{fig:Dij}
\end{figure}

Let $I_{D|G}^{*}(P;Q) = \max_G I_{D|G}(P;Q)$ be the largest possible mutual information achievable by an $\ell_x$-by-$\ell_y$ grid $G$ on a set $D$ of sample points. $P$ and $Q$ are the partitions of X-axis and Y-axis of grid $G$, respectively.  In order to have a fair comparison among different grids, the computed values of mutual information should be normalized. Since $I(P;Q)=H(Q)-H(Q|P)=H(P)-H(P|Q)$, we divide $I_{D|G}^{*}(P;Q)$ by $\log(\min(\ell_x,\ell_y))$. Therefore, we have
\begin{eqnarray}
 0\leq \frac{I_{D|G}^{*}(P;Q)}{\log(\min(\ell_x,\ell_y))} \leq 1.
\end{eqnarray}
This inequality motivates the definition of the MIC as a measure of dependency between two variables. For a dataset $D$ containing $n$ samples of two variables, we have
\begin{eqnarray}
 \text{MIC}(D) = \max_{\ell_x\ell_y<B(n)}\frac{I_{D|G}^{*}(P;Q)}{\log(min(\ell_x,\ell_y))}.
\end{eqnarray}
where $B(n)=n^{0.6}$ \cite{MIC:2011} or more generally $\omega(1)\leq B(n)\leq O(n^{1-\epsilon})$.
According to this definition, the MIC has the following properties:
\begin{itemize}
 \item $0\leq \text{MIC}(D)\leq 1$.
\item $\text{MIC}(x,y)=\text{MIC}(y,x)$.
\item It is invariant under order-preserving transformation applied to the dataset $D$.
\item It is not invariant under the rotation of coordinate axes, e.g., if $y=x$, then $\text{MIC}(D)=1$. However, after a $45^{\circ}$ clockwise rotation of coordinate axes, instead of $y=x$ we have $y=0$ and hence $\text{MIC}(D)=0$.
\end{itemize}

\subsection{MIC Algorithm}\label{subsec:MICALG}
Although the algorithm for computing the MIC is fully described in \cite{MIC:2011}, here we only review the \textit{OptimizeXAxis} algorithm which is used in computation of the highest mutual information achievable by an $\ell_x$-by-$\ell_y$ grid. Any $\ell_x$-by-$\ell_y$ grid imposes two sets of partitions on $x$-values (columns of grid) and $y$-values (rows of grid). We indicate columns of the grid by $\langle c_1,c_2,\ldots,c_{\ell_x}\rangle$ where $c_i$ denotes the endpoint (largest $x$-value) of the $i$-th column.

Since $I(P,Q)$ is upper-bounded by $H(P)$ and $H(Q)$, in order to maximize it, one can equipartition either the Y or X axis, i.e., impose a discrete uniform distribution on either $Q$ or $P$. Without loss of generality, we consider the version of the algorithm that equipartitions the Y-axis. However, it is obvious that we should check both of the cases (equipartitioning either the X or Y axis) separately for each $\ell_x$-by-$\ell_y$ grid and choose the maximum resulting mutual information. 

Let $H(P)$ denote the entropy of distribution imposed by $m$ sample points ($m\leq|D|=n$) on the partition of X-axis. Similarly, let $H(Q)$ denote the entropy of distribution imposed by $m$ sample points ($m\leq|D|=n$) on the partition of Y-axis. Since we have assumed that the Y-axis is equipartitioned, $H(Q)$ is constant and equal to $\log(|Q|)$. Finally, let $H(P,Q)$ denote the entropy of distribution imposed by $m$ sample points ($m<|D|=n$) on the cells of grid $G$ which has X-axis partition $P$ and Y-axis partition $Q$. Since $I(P;Q)=H(Q)-H(Q|P)$ and we have already maximized $H(Q)$ by equipartitioning the Y-axis, to achieve the highest mutual information, we have to minimize the $H(Q|P)$. This is done by the \textit{OptimizeXAxis} algorithm \cite{MIC:2011}.

An alternative formula for the mutual information is $I(P;Q)=H(Q)+H(P)-H(P,Q)$. Since $H(Q)$ is constant, the \textit{OptimizeXAxis} only needs to maximize $H(P)-H(P,Q)$. The following theorem \cite{MIC:2011} is the key to solve this problem.

\begin{figure}[t]
 \centering
 \includegraphics[width=45mm]{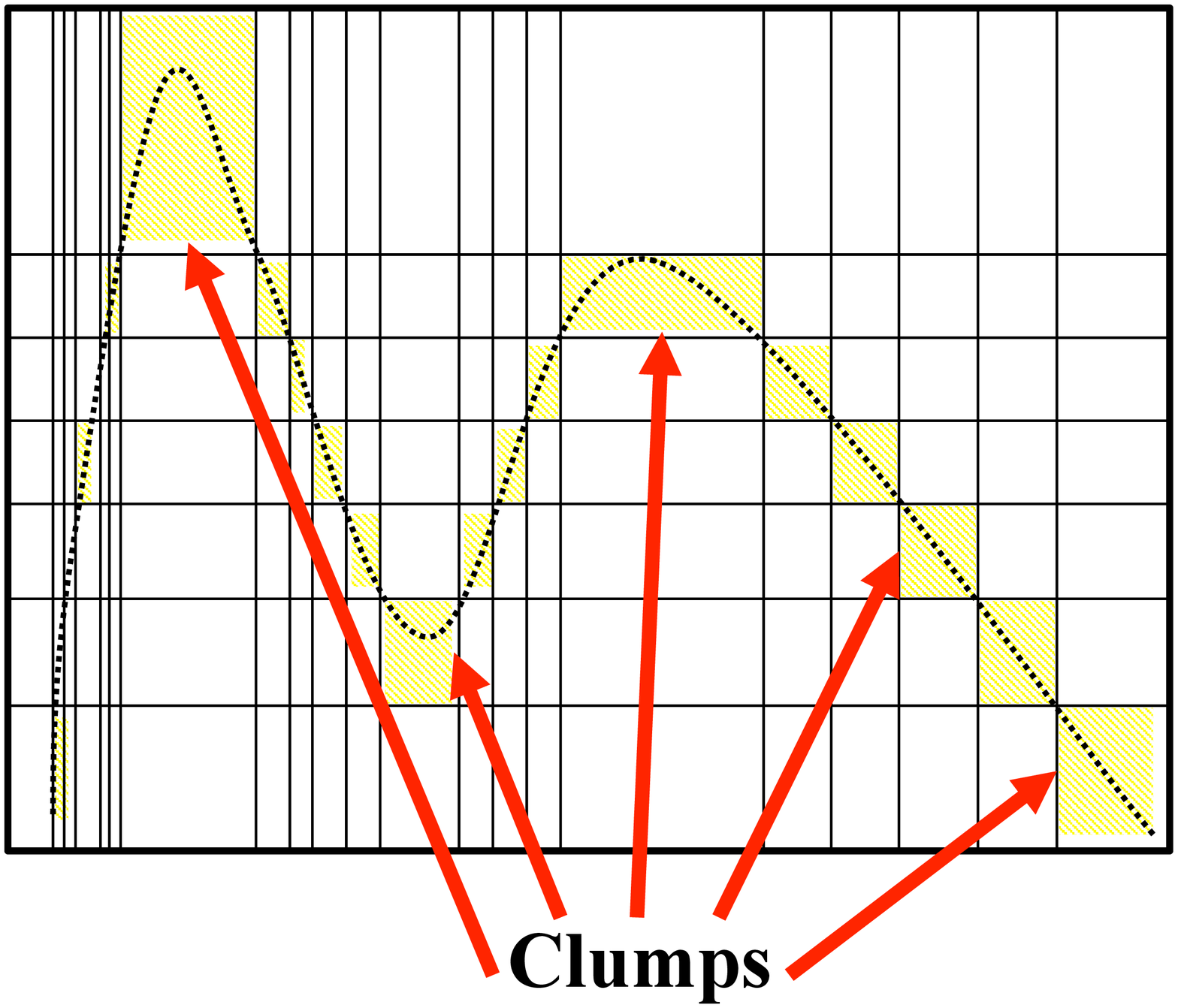}
 \caption{\textit{OptimizeXAxis} \cite{MIC:2011} considers only consecutive points falling into the same row and draw partitions between them. The set of consecutive points falling into the same row is called \textit{clump}.}
\label{fig:clump}
\end{figure}

\begin{theorem}
For a dataset $D$ of size $n$ and a fixed row partition $Q$, and for every $m,l\in \mathbb{N}$, if we define $F(m,l)=\max_{D(1:m),|P|=l}\{H(P)-H(P,Q)\}$ then for $l>1$ and $1 < m \leq n$ we would have the following recursive equation
\begin{eqnarray}
 F(m,l)=\max_{1\leq i<m}\{\frac{i}{m}F(i,l-1)-\frac{m-i}{m}H(\langle i,m \rangle,Q)\}.
\end{eqnarray}
\label{th:recursive} 
\end{theorem}
\begin{proof}[Proof of Theorem~\ref{th:recursive}]
 See proposition 3.2. in \cite{MIC:2011}.
\end{proof}
The \textit{OptimizeXAxis} uses dynamic programming technique motivated by Theorem \ref{th:recursive}. It ensures $F(n,l)$ that is the desired partition of dataset $D$ (which has $n$ sample points) having $l$ columns imposing partition $P$ over X-axis.
In order to minimize the $H(Q|P)$, \textit{OptimizeXAxis} considers only consecutive points falling into the same row and draw partitions between them. The set of consecutive points falling into the same row is called \textit{clump} (See Figure \ref{fig:clump} for a graphical view of clump). In Algorithm \ref{alg:OptimizeXAxis}, the GetClumpsPartition subroutine is responsible for finding and partitioning the clumps. Moreover, $P_{t,l}$ is an optimal partition of size $l$ for the first $t$ clumps.

\begin{algorithm}[!t]
\caption{OptimizeXAxis($D,Q,\ell_x$) \cite{MIC:2011} \label{alg:OptimizeXAxis}}
\textbf{Require:} $D$ is a set of ordered pairs sorted in increasing order by x-values

\textbf{Require:} $Q$ is a Y-axis partition of $D$

\textbf{Require:} $\ell_x$ is an integer greater than 1

\textbf{Ensure:} Returns a list of scores $(I_2,\ldots,I_{\ell_x})$ such that each $I_l$ is the maximum value of $I(P;Q)$ over all partitions $P$ of size $l$
\begin{algorithmic}[1]
\STATE $\langle c_0,\ldots,c_k \rangle \gets \text{GetClumpsPartition($D$,$Q$)}$
\STATE 
\STATE {Find the optimal partition of size 2}
\FOR {$t=2$ to $k$}
\STATE Find $s\in \{1,\ldots,t\}$ maximizing $H(\langle c_s,c_t\rangle)-H(\langle c_s,c_t\rangle,Q)$.
\STATE $P_{t,2} \gets \langle c_s,c_t \rangle$
\STATE $I_{t,2} \gets H(Q)+H(P_{t,2})-H(P_{t,2},Q)$
\ENDFOR
\STATE
\STATE {Inductively build the rest of the table of optimal partitions}
\FOR {$l=3$ to $\ell_x$}
\FOR {$t=2$ to $k$}
\STATE Find $s \in \{1,\ldots,t\}$ maximizing $F(s,t,l):=\frac{c_s}{c_t}(I_{s,l-1}-H(Q))+\sum_{i=1}^{|Q|}\frac{\#_{i,l}}{c_t}\log{\frac{\#_{i,l}}{\#_{*,l}}}$ 

where $\#_{*,j}$ is the number of points in the $j$-th column of $P_{s,l-1}\cup c_t$ and $\#_{i,j}$ is the number of points in the $j$-th column of $P_{s,l-1}\cup c_t$ that fall in the $i$-th row of $Q$
\STATE $P_{t,l} \gets P_{s,l-1} \cup c_t$
\STATE $I_{t,l} \gets H(Q)+H(P_{t,l})-H(P_{t,l},Q)$
\ENDFOR
\ENDFOR
\RETURN  $(I_{k,2},\ldots,I_{k,\ell_x})$
\end{algorithmic}
\end{algorithm}
\section{The Uniform-MIC (U-MIC)}\label{sec:sub-optimal}
\subsection{Noiseless Setting}
The major drawback of the Algorithm \ref{alg:OptimizeXAxis} is its computational complexity. If there exists $k$ clumps in the given partition of an $\ell_x$-by-$\ell_y$ grid, the runtime of this algorithm would be $O(k^2\ell_x \ell_y)$. If there is a functional association between the two variables, the number of clumps in the corresponding grid is pretty small. However, for noisy or random datasets it is easy to imagine that the number of clumps is very large and hence the computational complexity of the Algorithm {alg:OptimizeXAxis} would be large.

Furthermore and due to this problem, this algorithm cannot be generalized in order to detect associations between more than two variables. As an example, if we want to detect whether or not three variables are related to each other, we may write the formula for the \textit{generalized} mutual information as:
\begin{align}
  I(P;Q;R)=&H(P)+H(Q)+H(R)-H(P,Q)\\
&-H(P,R)-H(Q,R)+H(P,Q,R). \nonumber
 \end{align}
Hence, intuitively and like the case for two random variables, in order to maximize the generalized mutual information, we have to equipartition one axis to maximize the entropy. Nevertheless, we should partition the two other axes with respect to the places of the clumps in them. if we equipartition the first axis and there exists $k_1$ clumps in the second axis and $k_2$ clumps in the third axis, then the runtime of this algorithm would be $O(k_1^2k_2^2\ell_x\ell_yl_z^2)$ where $\ell_x,\ell_y,l_z$ are the sizes of partitioning. This runtime is not acceptable for large datasets. Therefore, we have to modify the algorithm in order to decrease its runtime and as a result make it generalizable to higher dimensions.

The algorithm we propose in here for replacing the Algorithm \ref{alg:OptimizeXAxis} is \textit{uniform partitioning} (Algorithm \ref{alg:UMIC}). Let $y_{\min}=\min_i y_i$, $y_{\max}=\max_j y_j$, and similarly $x_{\min}=\min_i x_i$ and $x_{\max}=\max_j x_j$. we then partition both $X$ and $Y$ axes such that all the columns have length $\frac{x_{\max}-x_{\min}}{\ell_x}$ and similarly all the rows have length $\frac{y_{\max}-y_{\min}}{\ell_y}$. We call this new measure, that is derived by replacing the Algorithm \ref{alg:OptimizeXAxis} with Algorithm \ref{alg:UMIC}, by the U-MIC (Uniform Maximal Information Coefficient).  In the following we prove that the U-MIC will approach 1 as the sample size grows for when there exists a functional association between two variables (with finite derivative). Without loss of generality, we do all the proofs in the case that $(x,y) \in [0,1]\times[0,1]$. These proofs could be generalized to other cases easily.

\begin{algorithm}
\caption{UniformPartition($\ell_x,\ell_y$)\label{alg:UMIC}}
\textbf{Require:} Dataset $D$ 

\textbf{Require:} $\ell_x$ and $\ell_y$ are integers greater than 1

\textbf{Ensure:} Returns a score $I^*$ which is the value of $I(P;Q)$ where $P$ are $Q$ are distributions from uniform partitioning of both axes.

\begin{algorithmic}[1]
\STATE $P \gets$ Uniform partition of $X$-axis by $\ell_x$ columns each has length $\frac{x_{\max}-x_{\min}}{\ell_x}$ 
\STATE $Q \gets$ Uniform partition of $Y$-axis by $\ell_y$ rows each has length $\frac{y_{\max}-y_{\min}}{\ell_y}$ 
\STATE $I^*=\frac{H(P)+H(Q)-H(P,Q)}{\log(\min(\ell_x,\ell_y))}$
\RETURN $I^*$
\end{algorithmic} 
\end{algorithm}

\begin{proposition}
If $D=\{(x_i,y_i)\}_{i=1}^n$ where $y_i=h(x_i)$ and $|h'(x)|<\infty$, then $\lim_{n\to \infty}\text{U-MIC}(D)=1$. \label{pro:UMIC}
\end{proposition}
\begin{proof}[Proof of Proposition \ref{pro:UMIC}]
We denote by $g_h(\alpha)$ the sub-level function of function $h(.)$, i.e., 
\begin{eqnarray}
 g_h(\alpha)=\lambda(\{x:h(x)\leq \alpha\}),
\end{eqnarray}
where $\lambda(\mathbb{T})$ denotes the fraction of sample points in the set $\mathbb{T}$. Consequently
\begin{eqnarray}
 g_h(\alpha)=F_y(\alpha)=\mathbb{P}(y\leq \alpha)=\mathbb{P}(h(x)\leq \alpha),
\end{eqnarray}
where $F_y(.)$ denotes the cumulative distribution function (CDF) and $\mathbb{P}(.)$ denotes the probability function. Using this notation and assuming that we uniformly partition $Y$-axis by $\ell_y$ rows, we can write the entropy of $Q$ which is the uniform partition of Y-axis as 
\begin{align}\label{eq:uniformH(y)}
 &H(Q) \\&=-\sum_{i=0}^{\ell_y-1}\mathbb{P}(Q=i)\log(P(Q=i))\nonumber \\
&=-\sum_{i=0}^{\ell_y-1}\mathbb{P}\left(\frac{i}{\ell_y}\leq Y < \frac{i+1}{\ell_y}\right)\log \left(\mathbb{P}\left(\frac{i}{\ell_y}\leq Y < \frac{i+1}{\ell_y}\right)\right)\nonumber\\
&=-\sum_{i=0}^{\ell_y-1}g'_h(\alpha_i)\frac{1}{\ell_y}\log \left(g'_h(\alpha_i)\frac{1}{\ell_y}\right)\nonumber\\
&=-\sum_{i=0}^{\ell_y-1}\frac{1}{\ell_y}g'_h(\alpha_i)\log(g'_h(\alpha_i))-\sum_{i=0}^{\ell_y-1}\frac{1}{\ell_y}g'_h(\alpha_i)\log\left(\frac{1}{\ell_y}\right)\nonumber,
\end{align}
where $\frac{i}{\ell_y} \leq \alpha_i < \frac{i+1}{\ell_y}$ for each $i$ ($0\leq i \leq \ell_y-1$) is derived according to the mean value Theorem. If without loss of generality we assume that $\min(\ell_x,\ell_y)=\ell_y$ then we can write
\begin{align}\label{eq:H(y) over ly}
 \frac{H(Q)}{\log(\ell_y)}=-\sum_{i=0}^{\ell_y-1}\frac{1}{\ell_y\log(\ell_y)}g'_h(\alpha_i)\log(g'_h(\alpha_i))+\sum_{i=0}^{\ell_y-1}\frac{1}{\ell_y}g'_h(\alpha_i).
\end{align}
As a result, in the asymptotic setting we can write
\begin{align}
 \lim_{\ell_y\to \infty}\frac{H(Q)}{\log(\ell_y)}=\sum_i\frac{1}{\ell_y}g'_h(\alpha_i)=1,
\end{align}
where the last equality holds since $ \lim_{\ell_y \to \infty}\sum_i\frac{1}{\ell_y}g'_h(\alpha_i)$ is the Riemann integral of the function $g'_h(\alpha_i)$.
If we assume that $|h'(.)|<c$, then according to the mean value Theorem we have 
\begin{eqnarray}\label{eq:upperbound-derivative}
 \left|h\left(\frac{i+1}{\ell_y}\right)-h\left(\frac{i}{\ell_y}\right)\right|\leq\frac{c}{\ell_y}.
\end{eqnarray}
Equation (\ref{eq:upperbound-derivative}) states that for a particular column of the X-axis partition, the curve of the function passes through at most $c+1$ cells of that column. We use this fact in upper-bounding the $H(Q|P)$. Similar to (\ref{eq:uniformH(y)}) and (\ref{eq:H(y) over ly}) we have 

\begin{align}\label{eq:uniformH(y|x)}
 \displaybreak[0]
 H(Q|P=k)&=-\sum_{i=0}^{\ell_y-1}\mathbb{P}(Q=i|P=k) \\
&~~~~~~~~~\times\log(\mathbb{P}(Q=i|P=k))\nonumber \\
&=-\sum_{i=0}^{\ell_y-1}\mathbb{P}\left(\frac{i}{\ell_y}\leq Y < \frac{i+1}{\ell_y}|P=k\right) \nonumber \\ 
&~~~~~~~~~\times\log\left(\mathbb{P}\left(\frac{i}{\ell_y}\leq Y < \frac{i+1}{\ell_y}|P=k\right)\right)\nonumber \\
\displaybreak[0]
&=-\sum_{i=0}^{\ell_y-1}f_{y|x}(\alpha_i|P=k)\frac{1}{\ell_y}  \nonumber \\ 
&~~~~~~~~~\times \log\left(f_{y|x}(\alpha_i|P=k)\frac{1}{\ell_y}\right)\nonumber \\
&=-\sum_{i=0}^{\ell_y-1}\frac{1}{\ell_y}f_{y|x}(\alpha_i|P=k) \nonumber \\
&~~~~~~~~~ \times \log(f_{y|x}(\alpha_i|P=k))\nonumber \\
&~~~~~~-\sum_{i=0}^{\ell_y-1}\frac{1}{\ell_y}f_{y|x}(\alpha_i|P=k)\log\left(\frac{1}{\ell_y}\right),\nonumber
\end{align}
where $f_{y|x}$ denotes the conditional probability density function. Because of equation (\ref{eq:upperbound-derivative}), we can simplify (\ref{eq:uniformH(y|x)}) as
\begin{align}
 H(Q|P=k)&=-\sum_{i=j_1}^{j_{c+1}}\frac{1}{\ell_y}f_{y|x}(\alpha_i|P=k) \\
&~~~~~~~~~\times \log(f_{y|x}(\alpha_i|P=k))\nonumber \\
&~~~~~~-\sum_{i=j_1}^{j_{c+1}}\frac{1}{\ell_y}f_{y|x}(\alpha_i|P=k)\log\left(\frac{1}{\ell_y}\right).\nonumber
\end{align}
If we define $k^{*} =\arg\max_{k}H(Q|P=k)$, then since $H(Q|P)=\sum_k\mathbb{P}(P=k)H(Q|P=k)$, we can write
\begin{align}
H(Q|P)&\leq-\sum_{i=j_1}^{j_{c+1}}\frac{1}{\ell_y}f_{y|x}(\alpha_i|P=k^*) \\
&~~~~~~~~~\times \log(f_{y|x}(\alpha_i|P=k^*)\nonumber \\
&~~~~~~-\sum_{i=j_1}^{j_{c+1}}\frac{1}{\ell_y}f_{y|x}(\alpha_i|P=k^*)\log\left(\frac{1}{\ell_y}\right),\nonumber
\end{align}
and hence
\begin{align}
 \lim_{\ell_y \to \infty}\frac{H(Q|P)}{\log(\ell_y)}\leq \sum_{i=j_1}^{j_{c+1}}\frac{1}{\ell_y}f_{y|x}(\alpha_i|P=k^*)=0.
\end{align}
The last equality holds since $\frac{1}{\ell_y} \to 0$ but $c<\infty$.
As a result
\begin{align}
\lim_{\ell_y\rightarrow\infty} \text{U-MIC}(D)&=\frac{I(P;Q)}{\log(\min\{\ell_x,\ell_y\})}\\
&=\frac{H(Q)-H(Q|P)}{\log(\min\{\ell_x,\ell_y\})}=1.\nonumber
\end{align}
\end{proof}

If $x$ and $y$ are independent, then according to the following Proposition we have $\text{U-MIC}(D)=0$.  

\begin{proposition}
If $D=\{(x_i,y_i)\}_{i=1}^n$ where $x_i \independent y_i$ for $1\leq i \leq n$, then $\text{U-MIC}(D)=0$.
\label{pro:UMIC_ind}
\end{proposition}
\begin{proof}[Proof of Proposition \ref{pro:UMIC_ind}]
The line of reasoning is straight forward and similar to the proof of Proposition \ref{pro:UMIC}. Since $x$ and $y$ are independent from each other, we can write 
\begin{align}
&H(Q) \\
&=-\sum_{i=0}^{\ell_y-1}\mathbb{P}(Q=i)\log(P(Q=i))\nonumber \\
&=-\sum_{i=0}^{\ell_y-1}\mathbb{P}\left(\frac{i}{\ell_y}\leq Y < \frac{i+1}{\ell_y}\right)\log \left(\mathbb{P}\left(\frac{i}{\ell_y}\leq Y < \frac{i+1}{\ell_y}\right)\right)\nonumber \\
&=-\sum_{i=0}^{\ell_y-1}\mathbb{P}\left(\frac{i}{\ell_y}\leq Y < \frac{i+1}{\ell_y}|P=k\right) \nonumber \\ 
&~~~~~~~~~~~\times\log\left(\mathbb{P}\left(\frac{i}{\ell_y}\leq Y < \frac{i+1}{\ell_y}|P=k\right)\right)\nonumber \\
&=-\sum_{i=0}^{\ell_y-1}\mathbb{P}(Q=i|P=k) \log(\mathbb{P}(Q=i|P=k))\nonumber \\
&=H(Q|P=k).\nonumber
\end{align}
Therefore, $H(Q)=H(Q|P=k)$ for every $k$ where $0\leq k\leq \ell_x-1$. Now since $H(Q|P)=\sum_k\mathbb{P}(P=k)H(Q|P=k)$, we have $H(Q)=H(Q|P)$ and as a result U-MIC($D$)=0.
\end{proof}

\subsection{Noisy Setting}
In this section we study performance of the U-MIC in noisy setting. We first give a lower-bound on it when the two variables $x$ and $y$ have a noisy functional association in which the noise is bounded. After that, we study the case of unbounded noise.

For the bounded noise case, without loss of generality we assume that $x\sim U[0,1]$ and the noise has a uniform distribution. Specifically, we assume that sample points $(x_i,y_i)$ have the form $(x_i,h(x_i)+z_\epsilon)$ where $z_\epsilon \sim U[-\epsilon,\epsilon]$. We define $y_{\rm mid} = \frac{y_{\max}+y_{\min}}{2}$. In Algorithm \ref{alg:UMIC}, we divide the Y-axis into two rows by drawing a horizontal line at $y_{\rm mid}$. In addition, we divide the X-axis into $\ell_x$ columns each having the length $\frac{1}{\ell_x}$ (since $x\sim U[0,1]$). Let $D_1 = \{ (x_i,y_i)|y_i<y_{\rm mid} \}$ and $D_2 = \{(x_i,y_i)|y_i>y_{\rm mid} \}$. We use $\mathcal{P}$ and $\mathcal{Q}$ to denote the partition of X-axis and Y-axis of the grid in this setting. Having this setting and notations in mind, the following Corollary gives a simple lower-bound for U-MIC($D$) in this case.

\begin{corollary}
 Let $m$ be the number of columns in $\mathcal{P}$ in which there exists a sample point $(\hat{x},\hat{y})$ such that $|\hat{y}-y_{\rm mid}|\leq \epsilon$. Then, U-MIC($D$) is lower-bounded by
\begin{eqnarray}
\frac{|D|\log(|D|)-|D_1|\log(|D_1|)-|D_2|\log(|D_2|)}{|D|}- \frac{m}{\ell_x}. \nonumber
\label{equ:boundUMIC}
\end{eqnarray}
\label{pro:UMICNoisy}
\end{corollary}
\begin{proof}[Proof of Corollary~\ref{pro:UMICNoisy}]
Since $I(\mathcal{P},\mathcal{Q})=H(\mathcal{Q})-H(\mathcal{Q}|\mathcal{P})$, we need to have an upper-bound on $H(\mathcal{Q}|\mathcal{P})$ in order to determine a lower-bound on $I(\mathcal{P},\mathcal{Q})$. According to the entropy definition we can write
\begin{align}\label{equ:entropy}
H(\mathcal{Q}) &= -\frac{|D_1|}{|D|}\log\left(\frac{|D_1|}{|D|}\right)-\frac{|D_2|}{|D|}\log\left(\frac{|D_2|}{|D|}\right) \\
&=\frac{|D|\log(|D|)-|D_1|\log(|D_1|)-|D_2|\log(|D_2|)}{|D|}. \nonumber
\end{align}
Let $\mathcal{M}=\{l_{p_1},\ldots,l_{p_m}\}$ denote the columns in which there exists a data point $(\hat{x},\hat{y})$ such that $|\hat{y}-y_{\rm mid}|\leq \epsilon$. Since $\mathcal{Q}$ has only two rows, we can upper-bound the $H(\mathcal{Q}|\mathcal{P})$ as the following
\begin{align}\label{equ:upp}
H(\mathcal{Q}|\mathcal{P})&=\sum_{k=0}^{\ell_x-1}\mathbb{P}(\mathcal{P}=k)H(\mathcal{Q}|\mathcal{P}=k) \\
&\overset{(a)}{=}\frac{1}{\ell_x}\left( \sum_{k \in \mathcal{M}} H(\mathcal{Q}|\mathcal{P}=k) +  \sum_{k \notin \mathcal{M}} H(\mathcal{Q}|\mathcal{P}=k)\right) \nonumber \\
&\overset{(b)}{=}\frac{1}{\ell_x} \sum_{k \in \mathcal{M}} H(\mathcal{Q}|\mathcal{P}=k) \nonumber \\
&\leq \frac{|\mathcal{M}|}{\ell_x}=\frac{m}{\ell_x}, \nonumber
\end{align}
where (a) holds since $x\sim U[0,1]$ and (b) holds because $z_\epsilon \sim U[-\epsilon,\epsilon]$. The lower-bound is then derived by combining \eqref{equ:entropy} and \eqref{equ:upp}.
\end{proof}

The main issue with generalizing this lower-bounding idea to other noise distributions is that noise values could be unbounded. Hence, we use the idea of \textit{$k$-nearest neighbors} to bound the noise so as to come up with a consistent version of the association detector. We study this idea for the case that noise is drawn from a Gaussian distribution with 0 mean and variance of $\sigma^2$.

For each sample point, we consider its $\delta_n$-neighborhood (we use subscript $n$ to show the dependency on the size of the dataset $n$). We replace each data point with the average of sample points located in its $\delta_n$-neighborhood. The following lemma characterizes the number of sample points in this neighborhood.

\begin{figure}[t]
 \centering
 \includegraphics[width=50mm]{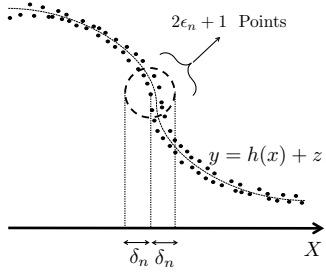}
 \caption{Using \textit{k-nearest neighbors} method to bound the noise in noisy relationships. We replace each point with the average of its neighbors in its $\delta_n$-neighborhood.}
\label{fig:kNearestNeighbors}
\end{figure}

\begin{lemma}
Let $x$ be uniformly distributed, i.e., $x \sim U[0,1]$ and $(x_i,y_i)$ denote the $i$-th data point in $D$ where $y_i=h(x_i)+z_i$. If $\mathcal{N}=\{(x_j,y_j)|(x_i-x_j)^2\leq \delta_n^2\}$, then $\lim_{n \to \infty}|\mathcal{N}|=2n\delta_n$.
\label{lem:neighbor}
\end{lemma}
\begin{proof}[Proof of Lemma~\ref{lem:neighbor}]
Let $\mathbb{I}(.)$ denote the indicator function. Then we can write
\begin{eqnarray}
2\epsilon_n=|\mathcal{N}|=\sum_{j=1}^{n}\mathbb{I}(x_i-\delta_n\leq x_j\leq x_i+\delta_n).
\end{eqnarray}

As a result $\mathbb{E}[2\epsilon_n]=2n\delta_n$. Using the Hoeffding inequality we have
\begin{eqnarray}
 \mathbb{P}(|2\epsilon_n-\mathbb{E}[2\epsilon_n]|\geq t) \leq 2e^{-2ct^2n^2},
\end{eqnarray}
for some constant $c$. If we let $t=\frac{1}{\log n}$, then $\lim_{n \to \infty}(\epsilon_n)=n\delta_n$ or equivalently,  $\lim_{n \to \infty}|\mathcal{N}|=2n\delta_n$.
\end{proof}
Assume that $h(.)$ is a Lipschitz continuous function of order $\beta$, i.e.,$~|h(v)-h(w)|\leq k|v-w|^{\beta}$ where $k$ is a constant depends on the function $h(.)$. If we estimate (or replace) the y-value of each noisy sample point with the average of sample points in its $\delta_n$-neighborhood, in the case of Gaussian noise (0 mean and variance of $\sigma^2$) we can write the estimation mean squared error as 
\begin{align}
 \Delta_n&=\frac{1}{n}\sum_{i=1}^n\mathbb{E}(\bar{h}(x_i)-h(x_i))^2\\
&=\frac{1}{n}\sum_{i=1}^n\mathbb{E}\left[\frac{\sum_{j=-\epsilon_n}^{\epsilon_n}(h(x_{i-j})+z_{i-j})}{2\epsilon_n+1}-h(x_i)\right]^2 \nonumber \\ &\leq \frac{k^2\epsilon_n^{2\beta}}{n^{2\beta}}+\frac{\sigma^2}{2\epsilon_n+1}. \nonumber
\end{align}
In order to minimize the estimation error we can take derivative with respect to $\epsilon_n$ and set it to 0. Therefore, $\epsilon_n$ which minimizes the mean squared error is 
\begin{eqnarray}
 \epsilon_n^*=\sqrt[2\beta+1]{\frac{\sigma^2}{4k^2\beta}}n^{1-\frac{1}{2\beta+1}}.
\label{equ:optimalNeighbors}
\end{eqnarray}
We use this $ \epsilon_n^*$ later to to bound the noise. The following lemma gives a probabilistic bound on the noise values.

\begin{lemma}
 If $z_1,z_2,\ldots,z_n$ are i.i.d. drawn from $N(0,\sigma^2)$, then $\mathbb{P}\{\max_{1\leq i \leq n} |z_i|>t\}\leq 2ne^{\frac{-t^2}{2\sigma^2}}$.
\label{lem:boundGaussian}
\end{lemma}
\begin{proof}[Proof of Lemma \ref{lem:boundGaussian}]
First of all, for a zero mean Gaussian random variable $z_i$ we prove that $\mathbb{P}\{ |z_i|>t\}\leq 2ne^{\frac{-t^2}{2\sigma^2}}$. Let $u=z_i-t$ and hence $u\thicksim N(-t,\sigma^2)$. We have
\begin{eqnarray}
\int_0^{\infty}\frac{1}{\sqrt{2\pi\sigma^2}}e^{\frac{-u^2-2ut}{2\sigma^2}}du\leq\int_0^{\infty}\frac{1}{\sqrt{2\pi\sigma^2}}e^{\frac{-u^2}{2\sigma^2}}du\leq1.
\end{eqnarray}
As a result we can write
\begin{eqnarray}
 \int_t^{\infty}\frac{1}{\sqrt{2\pi\sigma^2}}e^{\frac{-z_i^2}{2\sigma^2}}dz_i=\int_0^{\infty}\frac{1}{\sqrt{2\pi\sigma^2}}e^{\frac{-(u-t)^2}{2\sigma^2}}du\leq e^{\frac{-t^2}{2\sigma^2}}.
\label{equ:boundGaussian}
\end{eqnarray}
Similarly, (\ref{equ:boundGaussian}) holds for  $[-\infty,-t]$ and hence $\mathbb{P}\{ |z_i|>t\}\leq 2e^{\frac{-t^2}{2\sigma^2}}$. The result of lemma then follows from using union-bound on the $z_i$s. 
\end{proof}
By using the $k$-nearest neighbors method, each $z_i$ is replaced by $\bar{z}_i$ which is the average of $2\epsilon_n+1$ i.i.d. noise values and hence its variance is decreased by $2\epsilon_n+1$. This idea motivates the following corollary which lets us to bound the noise.
\begin{corollary}
By using the $k$-nearest neighbors method, $\bar{z}_i=\frac{\sum_{j=-\epsilon_n}^{\epsilon_n}z_{i-j}}{2\epsilon_n+1}$, and as a result $\lim_{n \to \infty}\max_{1\leq i \leq n}|\bar{z}_i| = 0$.
\label{Cor:limNoise}
\end{corollary}
\begin{proof}[Proof of Corollary \ref{Cor:limNoise}] 
According to the Lemma \ref{lem:boundGaussian}, we can write $\mathbb{P}\{\max_{1\leq i \leq n} |\bar{z}_i|>t\}\leq 2ne^{\frac{-t^2(2\epsilon_n+1)}{2\sigma^2}}$. The result then follows from letting $t=\frac{1}{\log n}$ and $\epsilon_n=\epsilon_n^*$ which was derived in (\ref{equ:optimalNeighbors}).
\end{proof}
In the next section we show how the U-MIC works in practice comparing to the MIC.
\section{Simulation Results}\label{sec:simul}

In this section, we study the performance of our proposed measure. We first show how it works for functional associations. Second, we study its performance for non-functional associations. Finally, we do some experiments for the case of noisy relationships. As mentioned previously, the authors in \cite{MIC:2011} apply a heuristic to compute the MIC which may not result in the true MIC. On the other hand, we do not apply any heuristic in the simulation results in order to have a precise comparison with our proposed method.

\begin{figure}[t]
\centering
\includegraphics[width=90mm]{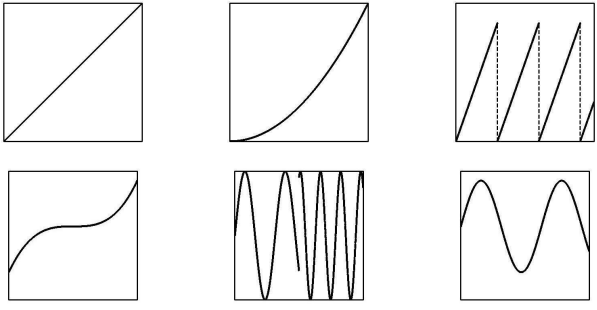}
\caption{Test functional relationships for Tables \ref{tab:funcs1},\ref{tab:funcs2}, and \ref{tab:timefuncs}.}
\label{fig:funcs}
\end{figure}

\begin{figure}[t]
\centering
\includegraphics[width=60mm]{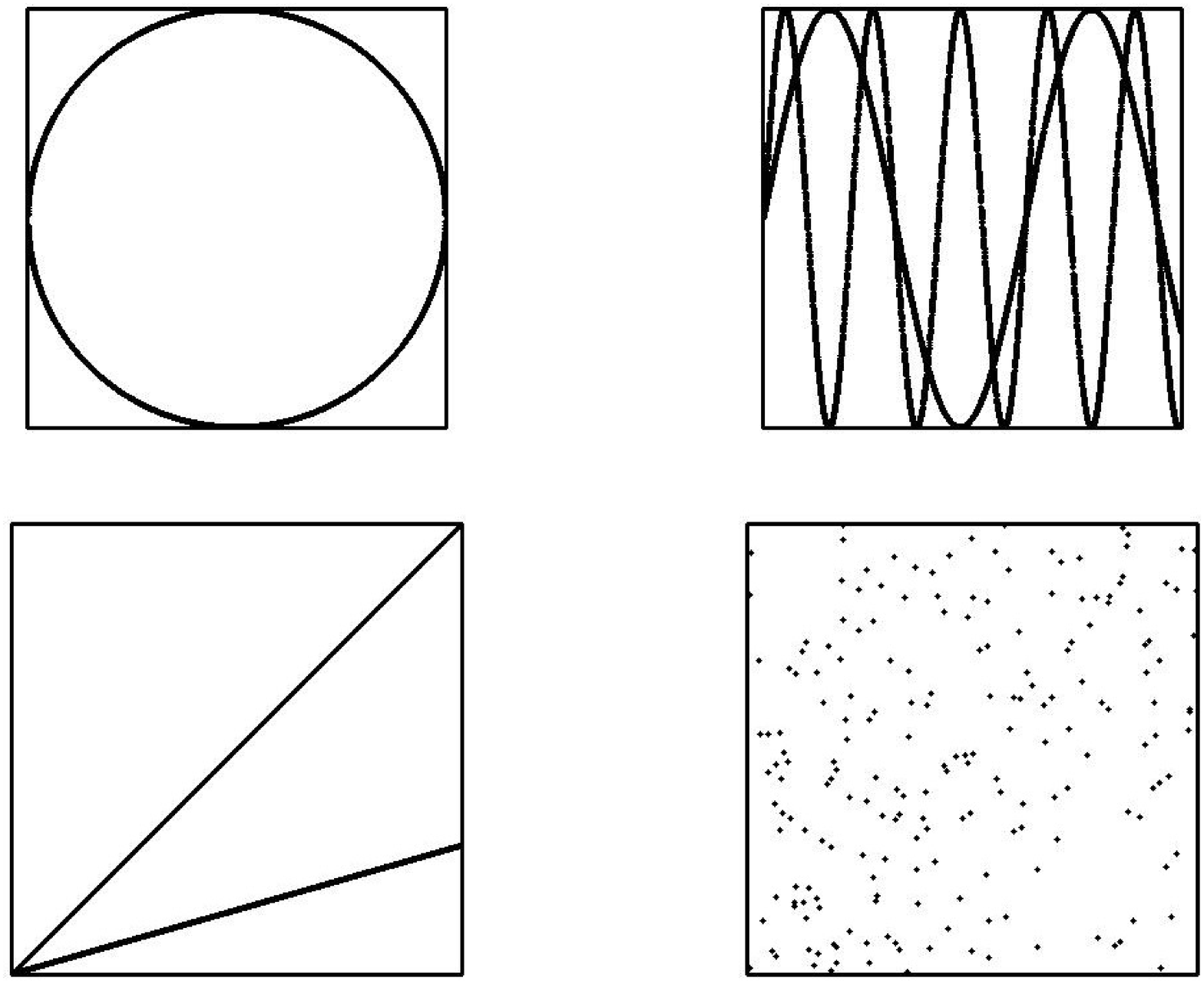}
\caption{Test non-functional relationships for Tables \ref{tab:non-funcs} and \ref{tab:timenon-funcs}. }
\label{fig:non-funcs}
\end{figure}

Figure \ref{fig:funcs} shows the functional associations that we have tested the performance of the MIC and U-MIC algorithms on. Table \ref{tab:funcs1} summarizes the results for the case that there are 200 sample points. One interesting point in Table \ref{tab:funcs1} is the value of the U-MIC for sinusoidal function with different frequencies. As we can see, MIC($D$)=1 while U-MIC($D$)=0.75 for this function. One interpretation of this difference is that in the proof of Proposition \ref{pro:UMIC}, we have assumed that the absolute value of derivative of function $h(.)$ is upper-bounded by constant $c$. However, this is not the case for sinusoidal function with different frequencies since there is a discontinuity in this function. If we increase the sample size, as reported in Table \ref{tab:funcs2}, this issue is alleviated as we can see.

\begin{table}[ht!]
\begin{center}
\resizebox{\columnwidth}{!}{
\begin{tabular}{|c || c  c  c  c  c  c | } 
 \hline
 &  Linear &  Parabolic & Periodic &  Cubic  &  Sin (Diff. Freq.) & Sin (Single Freq.) \\ [0.5ex] 
 \hline\hline
 \textbf{MIC}& 1 & 1 &1 & 1 & 1 & 1 \\ 
 \hline
 \textbf{U-MIC}& 1 &1 & 0.93 & 0.95 & 0.75  & 0.91 \\ 
 \hline
 \end{tabular}}
\caption{ MIC($D$) and U-MIC($D$) for different functional relationships in Figure \ref{fig:funcs}. For this set of experiments, $|D|=200$.}
\label{tab:funcs1}
\end{center}
\end{table}

\begin{table}[ht!]
\begin{center}
\resizebox{\columnwidth}{!}{
\begin{tabular}{|c || c  c  c  c  c  c | } 
 \hline
 &  Linear &  Parabolic & Periodic &  Cubic  &  Sin (Diff. Freq.) & Sin (Single Freq.) \\ [0.5ex] 
 \hline\hline
 \textbf{MIC}& 1 & 1 &1 & 1 & 1 & 1 \\ 
 \hline
 \textbf{U-MIC}& 1 &1 & 0.99 & 0.99 & 0.93  & 0.95 \\ 
 \hline
 \end{tabular}}
\caption{ MIC($D$) and U-MIC($D$) for different functional relationships in Figure \ref{fig:funcs}. For this set of experiments, $|D|=5000$.}
\label{tab:funcs2}
\end{center}
\end{table}

\begin{table}[ht!]
\begin{center}
\resizebox{\columnwidth}{!}{
\begin{tabular}{|c || c  c  c  c  c  c | } 
 \hline
 &  Linear &  Parabolic & Periodic &  Cubic  &  Sin (Diff. Freq.) & Sin (Single Freq.) \\ [0.5ex] 
 \hline\hline
 \textbf{MIC}& 0.1 & 0.5 &0.1 & 0.2 & 2 & 0.4 \\ 
 \hline
 \textbf{U-MIC}& 0.01 &0.01 & 0.01 & 0.01 & 0.01  & 0.01 \\ 
 \hline
 \end{tabular}}
\caption{ Run time (in sec.) for calculation of MIC($D$) and U-MIC($D$) for different functional relationships in Figure \ref{fig:funcs}. For this set of experiments, $|D|=200$.}
\label{tab:timefuncs}
\end{center}
\end{table}

Although the same issue holds for periodic function in Figure \ref{fig:funcs}, we do not see that much effect. Qualitatively, the derivative of continuous pieces of the periodic function in Figure \ref{fig:funcs} $(y=x)$ is smaller than the maximum of the derivative of sinusoidal function with different frequencies $(y=\sin(10x),~y=\sin(20x))$. Hence, if we uniformly partition the X-axis in the case of periodic function, there would be fewer sample points in rows of a certain column and more probably higher entropy (resulting in higher U-MIC), as the case in Table \ref{tab:funcs1}. 

Table \ref{tab:timefuncs} summarizes the runtime for calculation of the MIC and U-MIC for different functional associations in Figure \ref{fig:funcs}. As we can see, the U-MIC is at least 10 times faster in these cases. This is expected since the MIC uses dynamic programming to find a close to optimal grid for the data while the U-MIC just uniformly partitions the axes.

\begin{table}[ht!]
\begin{center}
\resizebox{\columnwidth}{!}{
\begin{tabular}{|c || c  c  c  c | } 
 \hline
 &  Circle &  Sinusoidal Mixture & Two Lines &  Random  \\ [0.5ex] 
 \hline\hline
 \textbf{MIC}& 0.68 &0.72 &0.71 & 0.16  \\ 
 \hline
 \textbf{U-MIC}& 0.64 &0.69 & 0.68 & 0.06 \\ 
 \hline
\end{tabular}}
\caption{ MIC($D$) and U-MIC($D$) for different non-functional relationships in Figure \ref{fig:non-funcs}. For this set of experiments, $|D|=200$.}
\label{tab:non-funcs}
\end{center}
\end{table}

\begin{table}[ht!]
\begin{center}
\resizebox{\columnwidth}{!}{
\begin{tabular}{|c || c  c  c  c | } 
 \hline
 &  Circle &  Sinusoidal Mixture & Two Lines &  Random  \\ [0.5ex] 
 \hline\hline
 \textbf{MIC}& 26.38 &13.41 & 16.60 & 61.00  \\ 
 \hline
 \textbf{U-MIC}& 0.01 & 0.02 &  0.01 & 0.02 \\ 
 \hline
\end{tabular}}
\caption{ Run time (in sec.) for calculation of MIC($D$) and U-MIC($D$) for different non-functional relationships in Figure \ref{fig:non-funcs}. For this set of experiments, $|D|=200$.}
\label{tab:timenon-funcs}
\end{center}
\end{table}

Table \ref{tab:non-funcs} summarizes the results for non-functional associations presented in Figure \ref{fig:non-funcs}. One important point about Table \ref{tab:non-funcs} is that the U-MIC has a better performance in the case of random sample points (i.e., $x \independent y$). In this case, the ideal MIC and U-MIC is 0; however, as we can see MIC($D$)=0.16 and U-MIC(D)=0.06. This issue is related to one of the criticisms made about MIC in the literature \cite{simon2014comment}. One of the drawbacks of the MIC is the fact that as a statistical test it has a lower power than other measures of dependency such as distance correlation \cite{simon2014comment}. In other words, it gives more false positives in detecting associations. However, according to our simulation results and Proposition \ref{pro:UMIC_ind} this issue is alleviated in the U-MIC.

Table \ref{tab:timenon-funcs} shows the runtime for calculation of the MIC and U-MIC. In the case of non-functional relationships we have more clumps in the initial grid of sample points for calculation of the MIC. Hence, Algorithm \ref{alg:OptimizeXAxis} which is basically running dynamic programming over initial grid to find the optimal grid, would have larger runtime as we can see in Table \ref{tab:timenon-funcs}. On the other hand, since the U-MIC is dealing with uniform partitioning of the grid of sample points, it does not matter what type of relationship the two random variables have. The runtime is almost constant and similar to the cases that there is a functional association between two variables.

Tables \ref{tab:noisy-non-funcs} and \ref{tab:noisy-timenon-funcs} summarize the results for noisy non-functional associations presented in Figure  \ref{fig:noisy-non-funcs}. Figure \ref{fig:noisy-non-funcs} is similar to Figure \ref{fig:non-funcs} except for the fact that we have added noise drawn from uniform distribution, i.e., $U[-0.5,0.5]$ to the sample points. Comparing Table \ref{tab:noisy-non-funcs}  with Table \ref{tab:non-funcs}, we can see that the range of decrease for different associations is almost the same for the both MIC and U-MIC. We expected this for MIC since it has an important property called equitability \cite{MIC:2011}. On the other hand, we can observe that at least according to the simulation results reported here, U-MIC has approximately the same equitability property.
\section{Conclusion}\label{sec:con}
In this paper we introduced a novel measure of dependency between two variables. This measure is called the uniform maximal information coefficient (U-MIC) because it is a modification of the original MIC \cite{MIC:2011}. It is derived from uniform partitioning of the both $X$ and $Y$ axes. Therefore, it is not dealing with dynamic programming similar to what the MIC does and hence, is much faster. We proved that asymptotically, U-MIC equals to 1 if there is a functional relationship between two variables. If two variables are truly independent from each other, then we showed that the U-MIC would be equal to 0. Specifically, according to the simulation results, we showed that the U-MIC does a better job in recognizing independence between variables comparing to the MIC.

\begin{figure}[t]
\centering
\includegraphics[width=80mm]{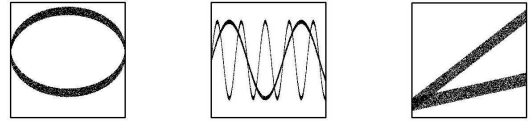}
\caption{Test noisy non-functional relationships for Tables \ref{tab:noisy-non-funcs} and \ref{tab:noisy-timenon-funcs}. }
\label{fig:noisy-non-funcs}
\end{figure}

\begin{table}[ht!]
\begin{center}
\resizebox{\columnwidth}{!}{
\begin{tabular}{|c || c  c  c  | } 
 \hline
 &  Circle &  Sinusoidal Mixture & Two Lines   \\ [0.5ex] 
 \hline\hline
 \textbf{MIC}& 0.54 &0.60 &0.57  \\ 
 \hline
 \textbf{U-MIC}& 0.52 &0.48 & 0.54 \\ 
 \hline
\end{tabular}}
\caption{ MIC($D$) and U-MIC($D$) for different non-functional relationships in Figure \ref{fig:non-funcs}. For this set of experiments, $|D|=200$ and noise is uniformly distributed in [-0.05,0.05].}
\label{tab:noisy-non-funcs}
\end{center}
\end{table}

\begin{table}[ht!]
\begin{center}
\resizebox{\columnwidth}{!}{
\begin{tabular}{|c || c  c  c  | } 
 \hline
 &  Circle &  Sinusoidal Mixture & Two Lines \\ [0.5ex] 
 \hline\hline
 \textbf{MIC}& 35 &16 & 27 \\ 
 \hline
 \textbf{U-MIC}& 0.01 & 0.02 &  0.01  \\ 
 \hline
\end{tabular}}
\caption{ Run time (in sec.) for calculation of MIC($D$) and U-MIC($D$) for different noisy non-functional relationships in Figure \ref{fig:non-funcs}. For this set of experiments, $|D|=200$ and noise is uniformly distributed in [-0.05,0.05].}
\label{tab:noisy-timenon-funcs}
\end{center}
\end{table}

\bibliographystyle{IEEETran}
\bibliography{sample_1.bib}
\end{document}